\documentclass[11pt]{article}

\usepackage[margin=1in]{geometry}
\usepackage{graphicx}
\usepackage[pdftex,colorlinks=true,linkcolor=blue,citecolor=blue,urlcolor=black]{hyperref}
\usepackage{amsmath, amsthm, amssymb}
\usepackage{comment}
\usepackage{url}
\usepackage{xcolor}
\usepackage{tikz}
\usepackage{subcaption}
\newcommand{\N}{\mathbb{N}}

\DeclareMathOperator{\dom}{dom}
\DeclareMathOperator{\range}{range}

\newcommand{\be}{\begin{equation}}
\newcommand{\ee}{\end{equation}}
\newcommand{\bea}{\begin{eqnarray}}
\newcommand{\eea}{\end{eqnarray}}
\newcommand{\bes}{\begin{equation*}}
\newcommand{\ees}{\end{equation*}}
\newcommand{\beas}{\begin{eqnarray*}}
\newcommand{\eeas}{\end{eqnarray*}}

\newtheorem{thm}{Theorem}
\newtheorem*{thm*}{Theorem}
\newtheorem{cor}[thm]{Corollary}

\newtheorem{fact}[thm]{Fact}
\newtheorem*{lem*}{Lemma}

\theoremstyle{definition}
\newtheorem{dfn}{Definition}

\begin{document}

\title{A composition theorem for decision tree complexity}

\author{Ashley Montanaro\footnote{{\tt am994@cam.ac.uk}}\\[11pt] {\small Centre for Quantum Information and Foundations, DAMTP, University of Cambridge, UK.}}

\maketitle

\begin{abstract}
We completely characterise the complexity in the decision tree model of computing composite relations of the form $h = g \circ (f^1,\dots,f^n)$, where each relation $f^i$ is boolean-valued. Immediate corollaries include a direct sum theorem for decision tree complexity and a tight characterisation of the decision tree complexity of iterated boolean functions.
\end{abstract}

% ------------------------------------------------------------------------------

\section{Introduction}

The decision tree model captures the complexity of computing functions $f:X^m \rightarrow Y$ in a setting where the quantity of interest is the number of queries to the input (see~\cite{buhrman02} for a good review of the model). We are allowed to query individual elements $x_i \in X$ at unit cost, and seek to compute $f(x_1,\dots,x_m)$ using the minimal number of queries.

Here we will be interested in the decision tree complexity of composite functions, i.e.\ functions of the form
\[ h(x) = g(f^1(x^1),\dots,f^n(x^n)), \]
where $x^i \in X^{m_i}$. One strategy to compute $h$ is first to replace $g$ with the function $\bar{g}$ given by substituting the values taken by any constant functions $f^i$ into $g$, and then to compute $\bar{g}$ using efficient algorithms for $f^1,\dots,f^n$ as black boxes. In other words, we use an efficient algorithm for computing $\bar{g}(y_1,\dots,y_n)$, and replace each query to a variable $y_i$ with the evaluation of $f^i(x^i)$ using an optimal algorithm. As the $x^i$ inputs are independent, it is natural to conjecture that this is in fact always the most efficient way to compute $g$. However, a moment's thought shows that this cannot be the case in general: indeed, consider the function $h(x) = g(f(x))$, where $f:\{0,1\}^n \rightarrow \{0,1\}^n$ is the identity function $f(x)=x$, and $g:\{0,1\}^n \rightarrow \{0,1\}$ is the function extracting the first bit of the input, $g(y) = y_1$. Then $f$ requires $n$ queries to be computed, whereas $h$ can be computed with 1 query. Another simple, but arguably less trivial, example of this phenomenon is illustrated in Figure \ref{fig:counter} below; the key point in this example is that the worst-case input to the functions $f^i$ may not produce an output which is a worst-case input to $g$.

We will show that, nevertheless, the simple algorithm above is optimal if the functions $f^i$ are {\em boolean-valued}, i.e.\ when $Y=\{0,1\}$. In fact, we prove such a result for the generalisation of the decision tree model to computing relations. A decision tree computes a relation $f \subseteq X^m \times Y$ if, on input $x \in X^m$, it outputs $y$ such that $(x,y) \in f$ (see Section~\ref{sec:defns} for a formal definition of the model). We call this result a composition theorem for decision tree complexity.

\begin{figure}
\begin{center}
\begin{subfigure}[b]{0.2\textwidth}
\begin{tikzpicture}[yscale=0.75,thick,every node/.style={inner sep=0mm,shape=circle}]
\node (1) at (0,3) {$x_1$};
\node (20) at (-1.25,2) {0};
\node (21) at (1.25,2) {$x_2$};
\node (310) at (0.5,1) {1};
\node (311) at (2,1) {2};

\draw (1) to (20); \draw (1) to (21); \draw (21) to (310); \draw (21) to (311);
\end{tikzpicture}
\caption{Optimal tree for $f$}
\end{subfigure}
\hspace{1cm}
\begin{subfigure}[b]{0.21\textwidth}
\begin{tikzpicture}[yscale=0.75,thick,every node/.style={inner sep=0mm,shape=circle}]
\node (1) at (0,3) {$y_1$};
\node (20) at (-1.25,2) {$y_2$};
\node (21) at (0,2) {1};
\node (22) at (1.25,2) {2};
\node (300) at (-2,1) {0};
\node (301) at (-1.25,1) {1};
\node (302) at (-0.5,1) {2};
\draw (1) to (20); \draw (1) to (21); \draw (1) to (22); \draw (20) to (300); \draw (20) to (301); \draw (20) to (302);
\end{tikzpicture}
\caption{Optimal tree for $g$}
\end{subfigure}
\hspace{1cm}
\begin{subfigure}[b]{0.25\textwidth}
\begin{tikzpicture}[yscale=0.75,thick,every node/.style={inner sep=0mm,shape=circle}]
\node (1) at (0,3) {$x_1$};
\node (20) at (-1.25,2) {$x_3$};
\node (21) at (1.25,2) {$x_2$};
\node (300) at (-2,1) {0};
\node (301) at (-0.5,1) {$x_4$};
\node (310) at (0.5,1) {1};
\node (311) at (2,1) {2};
\node (4010) at (-1,0) {1};
\node (4011) at (0,0) {2};
\draw (1) to (20); \draw (1) to (21); \draw (20) to (300); \draw (20) to (301); \draw (21) to (310); \draw (21) to (311); \draw (301) to (4010); \draw (301) to (4011);
\end{tikzpicture}
\caption{Optimal tree for $h$}
\end{subfigure}
\end{center}
\caption{Let $f:\{0,1\}^2 \rightarrow \{0,1,2\}$ and $g:\{0,1,2\}^2 \rightarrow \{0,1,2\}$ be defined by the decision trees above (where edges correspond to elements of $\{0,1\}$ or $\{0,1,2\}$ in ascending order from left to right). Set $h(x_1,x_2,x_3,x_4) = g(f(x_1,x_2),f(x_3,x_4))$. Then $f$ and $g$ require 2 queries each to be computed, but $h$ can be computed using only 3.}
\label{fig:counter}
\end{figure}
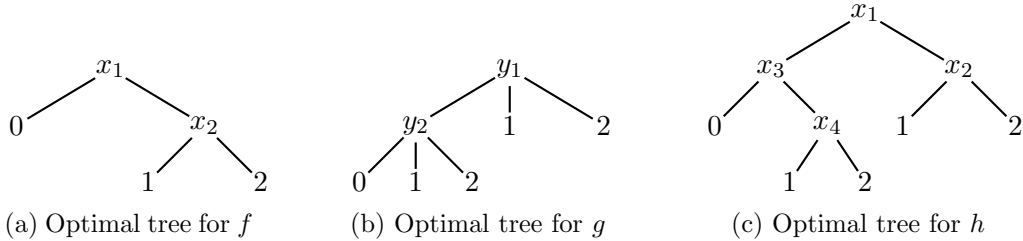

% ------------------------------------------------------------------------------

\subsection{Related work}

This work can be seen as fitting into the framework of ``direct sum'' results in query complexity, in which there has been significant recent interest (see the thesis~\cite{drucker12} for a good review). The basic aim of such work is to show that, given $n$ independent computational problems, the most efficient way to solve them all is simply to use an optimal algorithm for each problem separately. While this claim is intuitively obvious, it is not necessarily easy to prove, and in some cases is even false; see~\cite{drucker12} for some examples.

In particular, Jain, Klauck and Santha have shown~\cite{jain10} such a direct sum theorem for deterministic and randomised decision tree complexity. As discussed in~\cite{jain10}, previous work had dismissed this question as uninteresting, but the proof is not immediate. In the deterministic case, the result of~\cite{jain10} is proven by showing that, given a decision tree $T$ computing $n$ independent outputs, one can produce $n$ independent trees $T_i$, one for each output, and lower bounding the depth of $T$ in terms of the depth of the trees $T_i$. This approach does not seem suitable for proving a composition theorem, as given a tree computing $h$, it is not necessarily possible to produce individual trees for $f^1,\dots,f^n$. On the other hand, the composition theorem given here in fact implies a direct sum theorem for decision tree complexity (see Section \ref{sec:cor} for this and other corollaries).

Similar, but more complicated, composition theorems to the present work have been proven for quantities which are important in the study of {\em quantum} query complexity. H{\o }yer, Lee and \v{S}palek~\cite{hoyer05b} proved that the so-called non-negative weight adversary bound (which we will not define here; see~\cite{hoyer05b} for details) obeys a composition theorem. Interestingly, their proof is  based on generalising the query model to weighted queries, similarly to the approach taken here. A composition theorem was later proven for the general adversary bound (in two parts; a lower bound by H{\o }yer, Lee and \v{S}palek~\cite{hoyer07}, and an upper bound by Reichardt~\cite{reichardt09}), which was then used by Reichardt to infer a composition theorem for quantum query complexity~\cite{reichardt11}.

% ------------------------------------------------------------------------------

\subsection{Relations and decision trees}
\label{sec:defns}

We use essentially the standard model of decision trees (see e.g.~\cite{buhrman02} for definitions, and~\cite{jain10} for the extension to relations). Fix finite sets $X$ and $Y$, and integer $m$. A decision tree $T$ is a rooted $|X|$-ary tree whose internal vertices are labelled with indices between 1 and $n$, whose edges are labelled with elements of $X$, and whose leaves are labelled with elements of $Y$. $T$ computes a function $f_T:X^n \rightarrow Y$ as follows: starting with the root, it queries the input variable indexed by the label of the current vertex and, depending on the outcome, evaluates the corresponding subtree. When a leaf is reached, the tree outputs the label of that leaf. The depth of $T$ is the maximal length of a path from the root to a leaf (i.e.\ the worst-case number of queries used on any input).

We will consider the complexity of computing relations in the decision tree model. Let $f \subseteq X^n \times Y$. We write $\dom{f} = \{x:\exists y, (x,y) \in f\}$ and $\range{f} = \{y:\exists x, (x,y) \in f\}$, and also $f(x) = \{y:(x,y) \in f\}$, $f^{-1}(y) = \{x:(x,y) \in f\}$. $f$ is said to be a partial function if $|f(x)|=1$ for all $x \in \dom f$; in this case, we abuse notation and write $f(x)$ for the single element within the set $f(x)$. $f$ is a total function if in addition $\dom{f} = X^n$. We say that $f$ is {\em constant} if there exists $y \in Y$ such that $f^{-1}(y) = \dom f$, and {\em trivial} if, for all $y \in Y$, $f^{-1}(y) = \dom f$. Note that trivial relations are constant, and that the empty relation $f=\emptyset$ is trivial.

We say that a decision tree $T$ computes a relation $f$ if, for all $x \in \dom{f}$, on input $x$, $T$ outputs $y$ such that $(x,y) \in f$, i.e.\ $f_T(x) \in f(x)$ for all $x \in \dom{f}$. If this holds, we say that $f_T$ is consistent with $f$. The decision tree complexity $D(f)$ is the minimal depth over all decision trees computing $f$. This quantity is also known as the exact (classical) query complexity of $f$.

This can be generalised to define a weighted variant of decision tree complexity $D(f,[w_1,\dots,w_n])$, which is equal to the minimal worst-case number of queries required to compute $f$, given that querying the $i$'th element of the input costs $w_i$ queries. This can again be expressed as the minimal depth of a decision tree computing $f$ with weights $w_1,\dots,w_n$ on the inputs, where the depth of a decision tree $T$ with weights on the inputs is simply the maximal sum of the weights of the inputs specified by the labels of vertices of $T$, on any path from the root to a leaf. The usual decision tree complexity is recovered by setting $D(f) = D(f,[1,\dots,1])$.

For $f \subseteq X^n \times Y$, let $f_{i \leftarrow b} \subseteq X^n \times Y$ be the modification of $f$ we obtain by substituting $b$ as the $i$'th input to $f$. That is,
\[ f_{i \leftarrow b} = \{ (x,y) \in f : x_i = b \}. \]
We can use this definition to write down the following characterisation of weighted decision tree complexity (indeed, this could even be seen as its definition).

\begin{fact}
\label{fact:dt}
Let $f \subseteq X^n \times Y$, and $w_1,\dots,w_n \in \N$. Then $D(f,[w_1,\dots,w_n])=0$ if and only if $f$ is constant. If $f$ is not constant,
\[ D(f,[w_1,\dots,w_n]) = \min_{i \in [n]} \max_{b \in X} D(f_{i \leftarrow b},[w_1,\dots,w_n]) + w_i. \]
\end{fact}

\begin{proof}
If $f$ is constant, the tree can simply output $y$ such that $f^{-1}(y) = \dom{f}$ using no queries. On the other hand, if there exists a tree that computes $f$ and makes no queries, then there is only one fixed value $y$ that is output by the tree; so for all $x \in \dom{f}$, $(x,y) \in f$; so $f^{-1}(y) = \dom{f}$. If $f$ is not constant, then for any $i$, we can produce a tree computing $f$ with weights $w_1,\dots,w_n$ by first querying the $i$'th variable, obtaining outcome $b$, then using an optimal tree for $f_{i \leftarrow b}$ with weights $w_1,\dots,w_n$. This proves the ``$\le$'' part. For the ``$\ge$'' part, given a tree computing $f$ which makes its first query to the $i$'th variable, where $i$ achieves the minimum above, we can obtain a tree for computing $f_{i \leftarrow b}$, for any $b \in X$, by simulating this tree, replacing the first query with the constant $b$ and hence saving weight $w_i$.
\end{proof}

In the above fact, and the rest of the paper, we write $[n]$ for $\{1,\dots,n\}$.

% ------------------------------------------------------------------------------

\subsection{Composition of relations}

For $i \in [n]$, let $f^i \subseteq X^{m_i} \times Y$, and $g \subseteq Y^n \times Z$. Set $m = \sum_i m_i$. The composition $g \circ (f^1,\dots,f^n)$ is defined as
\[ \{(x,z) \in X^m \times Z: \exists y \in Y^n, \forall i\,(x^i,y_i) \in f^i \wedge (y,z) \in g \}, \]
where we write $x = (x^1,\dots,x^n)$, with $x^i \in X^{m_i}$. If $f$ and $g$ are total functions, this simply says that for all $x \in X^m$,
\[ h(x) = g( f^1(x^1), \dots, f^n(x^n) ). \]
In words, an algorithm to compute the relation $h$ on input $x$ has to output an arbitrary $z$ such that $z \in g(y)$, for some $y$ such that for all $i$, $y_i \in f^i(x^i)$.

Assume that $f^i$ is non-trivial for all $i$. If some of the relations $f^i$ are constant (i.e.\ $(f^i)^{-1}(b_i) = \dom{f^i}$, for some $b_i \in Y$), this expression for $h$ can be simplified by removing these inputs and modifying $g$ appropriately. Write $\bar{g}$ for the relation obtained by fixing constants to the inputs to $g$, where possible. In other words, if $S \subseteq [n]$ is the set of indices $i$ such that $f^i$ is not constant, then
\[ \bar{g} = \{ (y,z) \in g : y_i = b_i, i \notin S \}, \]
\[ h = \bar{g} \circ (f^1,\dots,f^n) = \{(x,z) \in X^m \times Z : \exists y \in Y^n, \forall i \in S, (x^i,y_i) \in f^i \wedge (y,z) \in \bar{g} \}.  \]

% ------------------------------------------------------------------------------

\section{A composition theorem}

We are finally ready to state and prove our main result.

\begin{thm}
\label{thm:main}
Fix $n \in \N$ and finite sets $X$, $Z$. Let $g \subseteq \{0,1\}^n \times Z$, and for $1 \le i \le n$, let $f^i \subseteq X^{m_i} \times \{0,1\}$ be non-trivial. Set $m = \sum_i m_i$. Define $h \subseteq X^m \times Z$ by
\[ h = g \circ (f^1,\dots,f^n). \]
Also let $\bar{g}$ be the relation given by substituting constant inputs into $g$. That is, for each $i$, if $(f^i)^{-1}(b_i) = \dom{f^i}$, then $g$ is replaced with $g_{i \leftarrow b_i}$ (as $f^i$ is non-trivial, there is exactly one such $b_i$). Then
\[ D(h) = D(\bar{g},[D(f^1),\dots,D(f^n)]). \]
\end{thm}

\begin{proof}
To see that $D(h) \le D(\bar{g},[D(f^1),\dots,D(f^n)])$, observe that the algorithm discussed in the introduction achieves this complexity. Formally, we compute $h$ as follows: first form a new relation $\bar{g}$ by replacing constant relations $f^i$ with corresponding constants as inputs to $g$; then use an optimal tree for computing $\bar{g}$ with weights $D(f^1),\dots,D(f^n)$,  replacing each query to the $i$'th variable with the use of an optimal tree for $f^i$. The more interesting part is the corresponding lower bound, the proof of which proceeds by induction on $D(h)$.

For the base case, take $D(h)=0$. Then $h$ is constant. By definition, this implies that there exists $z$ such that the following holds: For all $x = (x^1,\dots,x^n)$ such that there exist $y' \in \{0,1\}^n$ and $z' \in Z$ with $(x^i,y'_i) \in f^i$ for all $i$ and $(y',z') \in g$, there exists $y \in \{0,1\}^n$ such that $(x^i,y_i) \in f^i$ for all $i$ and $(y,z) \in g$. Let $S = \{i:f^i \text{ is not constant} \}$. For all $i \in S$, let $a^i,b^i \in X^{m_i}$ be picked such that $(a^i,0) \in f$, $(a^i,1) \notin f$, $(b^i,0) \notin f$ and $(b^i,1) \in f$. As $f^i$ is not constant, such elements $a^i$, $b^i$ exist. Further, for all $i \notin S$, let $c_i$ satisfy $(f^i)^{-1}(c_i) = \dom f^i$. As $f^i$ is not trivial, $c_i$ is unique and there exists $d^i \in X^{m_i}$ such that $(d^i,c_i) \in f^i$, but $(d^i,c'_i) \notin f^i$ for $c'_i \neq c_i$.

For $s:S \rightarrow \{0,1\}$, let the string $y(s) \in \{0,1\}^n$ be defined as follows: $y(s)_i = s(i)$ for $i \in S$, and $y(s)_i = c_i$ for $i \notin S$. For each distinct input $x \in \{0,1\}^m$ whose $i$'th block is equal to $a^i$ or $b^i$ (for $i \in S$) or $d^i$ (for $i \notin S$), a different string $y(s)$ is obtained, and if $(x,y) \in (f^1,\dots,f^n)$, $y=y(s)$. In particular, by choosing $s$ appropriately one can produce any $y \in \dom \bar{g}$. As there is exactly one $y$ such that $(x,y) \in (f^1,\dots,f^n)$, and $h$ is constant, $(y,z) \in g$. Thus, for all $y \in \dom \bar{g}$, $(y,z) \in \bar{g}$. Hence $\bar{g}$ is constant and $D(\bar{g},[D(f^1),\dots,D(f^n)]) = 0$ as required.

For the inductive step, assume that $D(h) = D(\bar{g},[D(f^1),\dots,D(f^n)])$ for all relations $h = g \circ (f^1,\dots,f^n)$ such that $f^i$ is non-trivial for all $i$, and such that $D(h) \le k$, for some $k$. Now let $h = g \circ (f^1,\dots,f^n)$ satisfy $D(h)=k+1$. As $h$ is not constant, by Fact \ref{fact:dt}
\beas
D(h) &=& \min_{i \in [m]} \max_{b \in X} D(h_{i \leftarrow b}) + 1\\
&=& \min_{i \in [n]} \min_{j \in [m_i]} \max_{b \in X} D(g \circ (f^1,\dots,f^i_{j \leftarrow b},\dots,f^n) ) + 1.
\eeas
As $f^i$ is non-trivial for all $i$, for any pair $i,j$ there exists $b'$ such that $f^i_{j \leftarrow b'}$ is non-trivial. Thus, for $b$ achieving the above maximum, $f^i_{j \leftarrow b}$ is non-trivial. Let $i,j$ be the indices achieving the above minimum. Then, by the inductive hypothesis,
\[ \max_{b \in X} D(g \circ (f^1,\dots,f^i_{j \leftarrow b},\dots,f^n) ) = \max_{b \in X} D(\bar{\bar{g}},[D(f^1),\dots,D(f^i_{j \leftarrow b}),\dots,D(f^n)] ), \]
where $\bar{\bar{g}}$ is the relation obtained by substituting $b_k$ as the $k$'th input to $g$, for all $k \neq i$ such that $(f^k)^{-1}(b_k) = \dom f^k$, and also substituting $b_i$ as the $i$'th input to $g$ if $(f^i_{j \leftarrow b})^{-1}(b_i) = \dom f^i_{j \leftarrow b}$. We thus have
\beas
D(h) &=& \max_{b \in X} D(\bar{\bar{g}},[D(f^1),\dots,D(f^i_{j \leftarrow b}),\dots,D(f^n)] ) + 1\\
&=& D(\bar{\bar{g}},[D(f^1),\dots,\max_{b \in X} D(f^i_{j \leftarrow b}),\dots,D(f^n)] ) + 1,
\eeas
where the second equality holds because decreasing any of the weights cannot increase the decision tree complexity. First assume that, for $b$ achieving the maximum, $f^i_{j \leftarrow b}$ is not constant. Then $\bar{\bar{g}} = \bar{g}$ and
\beas
D(h) &=& D(\bar{g},[D(f^1),\dots,\max_{b \in X} D(f^i_{j \leftarrow b}),\dots,D(f^n)] ) + 1\\
&\ge& D( \bar{g},[D(f^1),\dots,D(f^i) - 1,\dots,D(f^n)] ) + 1\\
&\ge& D( \bar{g},[D(f^1),\dots,D(f^n)] )
\eeas
as required. The first inequality follows from $D(f^i) \le \max_{b\in X} D(f^i_{j \leftarrow b}) + 1$, while the second  holds because, to find a decision tree $T$ which computes $\bar{g}$ with weight $D(f^i)$ on the $i$'th variable, we can just use an optimal decision tree $T'$ for $\bar{g}$ with weight $D(f^i) - 1$ on the $i$'th variable. As $T'$ is optimal, we only encounter the $i$'th variable at most once on any path from the root to a leaf, so the depth of $T$ is upper bounded by the depth of $T'$, plus 1.

On the other hand, now assume that, for $b$ achieving the maximum, $f^i_{j \leftarrow b}$ is constant. Then it must hold that $D(f^i)=1$ (if $f^i$ were already constant, querying any of the variables on which it depends would not achieve the overall minimum; if $f^i$ is not constant, and there exists $b' \in X$ such that $f^i_{j \leftarrow b'}$ is not constant, such a $b'$ would achieve the maximum over $b$; hence $D(f^i)=1$). So we can produce a decision tree for $\bar{g}$ with weights $D(f^1),\dots,D(f^n)$ by computing $f^i$ using one query to the $j$'th variable and then using an optimal decision tree for $\bar{\bar{g}}$ with weights $D(f^1),\dots,D(f^{i-1}),0,D(f^{i+1}),\dots,D(f^n)$. Thus
\[ D( \bar{g},[D(f^1),\dots,D(f^n)] ) \le \max_{b \in X} D(\bar{\bar{g}},[D(f^1),\dots,D(f^i_{j \leftarrow b}),\dots,D(f^n)] ) + 1 = D(h) \]
as required, completing the proof.
\end{proof}

% ------------------------------------------------------------------------------

\subsection{Corollaries}
\label{sec:cor}

We finish by noting a few simple corollaries of Theorem \ref{thm:main} that can be useful in amplifying separations between decision tree complexity and other complexity measures. We begin by considering the case where all the relations $f^i$ are the same.

\begin{cor}
Fix $f \subseteq X^k \times \{0,1\}$ and $g \subseteq \{0,1\}^n \times Z$. Define $h \subseteq X^{nk} \times Z$ as $h = g \circ (f,f,\dots,f)$. Then $D(h) = D(f) D(g)$.
\end{cor}

\begin{proof}
If $f$ is constant, then $D(h) = 0$. Otherwise, take $f^1 = \dots = f^n = f$ in Theorem \ref{thm:main}. Then $D(h) = D(g,[D(f),\dots,D(f)]) = D(f) D(g)$.
\end{proof}

Second, we observe that the composition theorem can be used to characterise the decision tree complexity of an iterated boolean function, via the following corollary.

\begin{cor}
\label{cor:iter}
For any $f:\{0,1\}^n \rightarrow \{0,1\}$, let the $k$-fold iteration of $f$, $f^{(k)}:\{0,1\}^{n^k} \rightarrow \{0,1\}$, be defined as follows. Set $f^{(1)}=f$, and for $k \ge 2$, split the input $x \in \{0,1\}^{n^k}$ into blocks $x^1,\dots,x^n$ of $n^{k-1}$ bits each, and set $f^{(k)}(x) = f( f^{(k-1)}(x^1),\dots,f^{(k-1)}(x^n) )$. Then $D(f^{(k)}) = D(f)^k$.
\end{cor}

\begin{proof}
If $f$ is constant, then $D(f^{(k)}) = 0 = D(f)^k$. Otherwise, by Theorem \ref{thm:main}, for $k \ge 2$, $D(f^{(k)}) = D(f, [D(f^{(k-1)}),\dots,D(f^{(k-1)})]) = D(f) D(f^{(k-1)})$. The claim follows by induction.
\end{proof}

Nisan and Szegedy used an iterated version of the ``not all equal'' function on 3 bits to prove an asymptotic separation between decision tree complexity and polynomial degree~\cite{nisan94}, based on lower bounding decision tree complexity by sensitivity. Corollary~\ref{cor:iter} gives a direct proof of this result.

Finally, we obtain the following special case of a direct sum result of Jain, Klauck and Santha~\cite{jain10}.

\begin{cor}
For any non-trivial $f^1,\dots,f^n \subseteq X^k \times \{0,1\}$, let $h \subseteq X^{kn} \times \{0,1\}^n$ be defined by $h = (f^1,f^2,\dots,f^n)$. Then $D(h) = \sum_{i \in [n]} D(f^i)$.
\end{cor}

\begin{proof}
Take $g(y_1,\dots,y_n) = (y_1,\dots,y_n)$ in Theorem \ref{thm:main}. Then $D(h) = D(\bar{g},[D(f^1),\dots,D(f^n)])$. Let $S = \{i:f^i \text{ is not constant} \}$. Then $D(\bar{g},[D(f^1),\dots,D(f^n)]) = \sum_{i \in S} D(f^i) = \sum_{i \in [n]} D(f^i)$.
\end{proof}

The proof technique of~\cite{jain10} allows the range of each relation $f^i$ to be arbitrary, rather than being restricted to $\{0,1\}$ as here.

% ------------------------------------------------------------------------------

\subsection*{Acknowledgements}

I would like to thank Graeme Mitchison and Toby Cubitt for helpful discussions and comments on this work, and an anonymous referee for inspiring it.

% ------------------------------------------------------------------------------

\bibliographystyle{plain}
\bibliography{../thesis}

% ------------------------------------------------------------------------------
\end{document}